\documentclass{article}
\usepackage[T1]{fontenc}
\usepackage{graphicx}
\usepackage{times}
\usepackage{authblk}
\usepackage{soul}
\usepackage{url}
\usepackage[hidelinks]{hyperref}
\usepackage[utf8]{inputenc}
\usepackage[small]{caption}
\usepackage{graphicx}
\usepackage{amsthm}
\usepackage{amsmath}
\usepackage{amssymb}
\usepackage{booktabs}
\usepackage{algorithm}
\usepackage{algorithmic}
\usepackage[switch]{lineno}

\newtheorem{theorem}{Theorem}
\newtheorem{conjecture}{Conjecture}
\newtheorem{definition}{Definition}
\newtheorem{lemma}{Lemma}
\newtheorem{corollary}{Corollary}
\usepackage{thm-restate}

\usepackage{newfloat}
\usepackage{listings}
\usepackage{diagbox}

\DeclareCaptionStyle{ruled}{labelfont=normalfont,labelsep=colon,strut=off} 
\floatstyle{ruled}
\newfloat{listing}{tb}{lst}{}
\floatname{listing}{Listing}
\newcommand{\kSNM}[1]{#1\textsc{-SNM}}
\newcommand{\SNM}{\textsc{SNM}}

\newcommand{\kSNMx}[2]{#1\textsc{-SNM-}#2}

\newcommand{\RSEB}{\textsc{RSEB}}
\newcommand{\RKotH}{\textsc{RKotH}}
\newcommand{\SigOnly}{\textsc{SignificantOnly}}
\newcommand{\RdSEB}[1]{\textsc{R}#1\textsc{SEB}}

\usepackage[shortlabels]{enumitem}
\usepackage{xcolor}

\newcommand{\Sn}[1]{\mathcal{S}_{#1}}

\newcommand{\pr}[1]{\text{Pr}[{#1}]}
\newcommand{\cpr}[2]{\text{Pr}[{#1} \mid {#2}]}

\title{On Approximately Strategy-Proof Tournament Rules for Collusions of Size at Least Three}

\author[1]{David Mik\v{s}an\'{i}k \footnote{dav.miksanik@gmail.com. This research is part of a project that has received funding from the European Union’s Horizon 2020 research and innovation programme under the Marie Skłodowska-Curie grant agreement No. 823748, and while this author was a participant in the DIMACS REU program at Rutgers University, supported by NSF grant CNS-2150186.}}
\author[2]{Ariel Schvartzman \footnote{aschvartzman@google.com}}
\author[1]{Jan Soukup \footnote{soukup@kam.mff.cuni.cz. This research is part of a project that has received funding from the European Union’s Horizon 2020 research and innovation programme under the Marie Skłodowska-Curie grant agreement No. 823748, and while this author was a participant in the DIMACS REU program at Rutgers University, supported by NSF grant CNS-2150186.}}

\affil[1]{Computer Science Institute, Charles University, Prague, Czechia}
\affil[2]{Google Research, Mountain View, California, USA}

\begin{document}

\maketitle

\begin{abstract}
A tournament organizer must select one of $n$ possible teams as the winner of a competition after observing all $\binom{n}{2}$ matches between them. The organizer would like to find a tournament rule that simultaneously satisfies the following desiderata. It must be \emph{Condorcet-consistent} (henceforth, CC), meaning it selects as the winner the unique team that beats all other teams (if one exists). It must also be \emph{strongly non-manipulable} for groups of size $k$ at probability $\alpha$ (henceforth, $\kSNMx{k}{\alpha}$), meaning that no subset of $\leq k$ teams can fix the matches among themselves in order to increase the chances any of it's members being selected by more than $\alpha$. Our contributions are threefold. First, wee consider a natural generalization of the Randomized Single Elimination Bracket rule from~\cite{RSEB} to $d$-ary trees and provide upper bounds to its manipulability. Then, we propose a novel tournament rule that is CC and $\kSNMx{3}{1/2}$, a strict improvement upon the concurrent work of~\cite{deathmatch} who proposed a CC and $\kSNMx{3}{31/60}$ rule. Finally, we initiate the study of reductions among tournament rules. 
\end{abstract} 

\section{Introduction}
\label{sec:intro}

Consider the problem a tournament organizer faces when, after observing all pairwise matches between $n$ teams, they must select one as the winner of the tournament. We model the tournament $T$ as a complete, directed graph on the $n$ teams. A tournament rule $r$ is a (possibly randomized) mapping from the set of tournaments on $n$ teams $\mathcal{T}_n$ to a probability vector in $\Delta^n$, $r: \mathcal{T}_n \rightarrow \Delta^n$. The tournament organizer is thus tasked with designing a tournament rule $r$ and would like the rule to satisfy the following natural properties: 

\begin{enumerate}
    \item \label{prop1} If there is a team who beats all other teams, termed a \emph{Condorcet-winner}, they should be picked as the winners of the tournament with probability $1$. We call such rules \emph{Condorcet-consistent} (or CC). 
    \item \label{prop2} No team should be incentivized to unilaterally throw their own games in order to obtain a better outcome. We call such rules \emph{monotone}.
    \item \label{prop3} No subset of $\leq k$ teams should have incentives to fix the matches among themselves in order to improve the chances of any of its members being selected as the winner of the tournament. We call such rules $k$-\emph{strongly non-manipulable} (or \kSNM{k}). 
\end{enumerate} 

These properties are motivated by real-world sports competitions. It would be \newline unimaginable to violate Property~\ref{prop1} and not award the top prize to an undefeated team. Violations to Properties~\ref{prop2},~\ref{prop3} have been observed in high-stakes competitions such as the Olympic Games and the FIFA World Cup. An infamous scandal in the Women's Doubles Badminton tournament at the 2012 Olympics saw multiple teams purposefully losing the last games of their group stage matches in order to avoid a difficult match-up in the following single-elimination bracket. This clear violation of monotonicity (and sportsmanship) resulted in the disqualification of 4 teams, including many of the likely medalists. A less investigated but equally egregious scandal occurred during the 1982 FIFA World Cup. West Germany and Austria disputed the last match of their group stage with full knowledge of the outcomes of all other games. It is suspected that they colluded in order to produce an outcome that would see both teams advance to the next stage of the tournament, at the expense of Algeria, who unexpectedly defeated West Germany in their opening match. As a result of this possible violation of strong non-manipulability, the last game of every group in every FIFA World Cup since has been played simultaneously. 

Observe that if one wanted to simply satisfy Properties~\ref{prop2},~\ref{prop3}, there are numerous simple rules that do so. For example, picking a winner uniformly at random, picking a fixed team as the winner (i.e., a dictatorship) or picking the winner proportional to the number of wins in the tournament all satisfy Properties~\ref{prop2},~\ref{prop3} but not Property~\ref{prop1}. Similarly, it is easy to satisfy Properties~\ref{prop1},~\ref{prop2}. If there is a Condorcet-winner, pick that team. Otherwise, pick a team uniformly at random. Unfortunately, it is known that Properties~\ref{prop1},~\ref{prop3} are directly at odds with each other: \cite{nonmanipulable_randomized_tournament_selections} showed that there exists no randomized tournament rule that can satisfy both of these properties at the same time, even for $k=2$. One way to overcome this impossibility result is to relax Property~\ref{prop3} as follows:
\begin{enumerate}[resume]
    \item \label{prop3r} No subset of $\leq k$ teams should be able to fix the matches among themselves in order to improve the chances of any of its members being selected as the winner by more than $\alpha$. We call such rules $\kSNMx{k}{\alpha}$.
\end{enumerate}
A growing body of work has asked what is the smallest $\alpha$ for which there exists a rule that satisfies Properties~\ref{prop1},~\ref{prop2} and ~\ref{prop3r} (for some fixed value of $k$). First, the work of~\cite{RSEB} proves that the Random Single-Elimination Bracket (henceforth $\RSEB$) rule is CC and $\kSNMx{2}{1/3}$, and that no other CC rule can do better. Later,~\cite{LPmanipulability} show that a rule termed Randomized King-of-the-Hill (henceforth, \RKotH) matches the performance of $\RSEB$ and satisfies a condition even stronger than CC. These works completely settle the question of finding CC and minimally manipulable (or \emph{optimal}) strategy-proof rules for collusions of size $k = 2$. On the other hand, very little is known about the case when $k > 2$, even $k = 3$.~\cite{RSEB} prove a simple lower bound of $\alpha \geq \frac{k-1}{2k-1}$.~\cite{LPmanipulability} proved that there exists an LP-based rule that is CC and $\kSNMx{k}{2/3}$ \emph{for all} $k$ simultaneously. Unfortunately, this rule is neither explicit nor monotone. More recently, in concurrent work~\cite{deathmatch} gave the first explicit, CC, monotone, $\kSNMx{3}{\alpha}$ rule for $\alpha = 31/60$ (and this value of $\alpha$ is tight for their rule). 

As hinted in the previous paragraph, there are two approaches to proving the existence of CC and approximately strategy-proof tournament rules. One approach takes simple rules (such as \RSEB, \RKotH), in the hopes that they are not too manipulable, and provides tight analysis for them. This does not always work: many simple rules, such as picking the team with the most wins, are extremely manipulable (i.e., have $\alpha = 1-O(1/n)$, see~\cite{RSEB}). The other approach provides rules which are not explicitly implementable. For example, the LP-based rule of~\cite{LPmanipulability} arises from fixing the tournament rule for tournaments with a Condorcet-winner, relaxing the manipulability constraints for tournaments that are close to having a Condorcet-winner and proving that the resulting polytope is non-empty for some value of $\alpha < 1$. Our results make use of both of these approaches and introduce another. 

Our first contribution, inspired by the positive results of~\cite{RSEB}, introduces a natural generalization of $\RSEB$. The $\RSEB$ rule randomly seeds teams on the leaf nodes of a binary tree and recursively labels inner nodes as the winner of the match between its children. The winner of the bracket is the team whose label appears in the root node of the tree. We define the Random $d$-ary Single Elimination Bracket (henceforth $\RdSEB{d}$) rule similarly with one key difference: instead of using binary trees, we use $d$-ary trees. If the sub-tournament induced by an inner node's children has a Condorcet-winner, then the inner node will carry the label of that child. Otherwise, the inner node picks a child uniformly at random to advance~\footnote{This decision is inspired by the observation that tournaments on three teams either have a Condorcet-winner or have three teams that beat each other cyclically.}. We provide an upper bound $\alpha_{d, k} < 1$ on the manipulability of $\RdSEB{d}$ on tournaments with $n$ teams and collusions of size up to $k \leq d$.  

\begin{restatable}{theorem}{thmRdSEB}
\label{thm:RdSEB}
        Let $2 \leq k \leq d$. The  $\RdSEB{d}$ rule is Condorcet-consistent, monotone and $\kSNMx{k}{\alpha_{d, k}}$ for 
        
        \begin{equation*}
        \alpha_{d, k} \leq 1 - \Big(\frac{2 \cdot (d)_k }{d^{k+1}} \Big),
        \end{equation*}
        where $(d)_k = \prod_{i=0}^{k-1} (d-i)$ is the falling factorial of $d$ with $k$ terms. 
\end{restatable}

For $k = d = 3$, we obtain that $\alpha_{3, 3} = .8519$. As a consequence of Theorem~\ref{thm:RdSEB}, we get the first explicit family of CC, monotone rules whose manipulability for any (fixed) $k$ is bounded away from $1$. This stands in contrast to the LP-based rule of~\cite{LPmanipulability} which was neither monotone nor explicit and to several of the rules analyzed in~\cite{RSEB} which had $\alpha \rightarrow 1$ as $n \rightarrow \infty$, even for $k = 2$. In other words, the bound from Theorem~\ref{thm:RdSEB} is independent of $n$, the number of competing teams.  

Our second contribution, inspired more by the second approach to finding approximately optimal tournament rules, is a new explicit tournament rule which strictly improves upon the results of~\cite{deathmatch}. 

\begin{restatable}{theorem}{thmsnm}
\label{thm:main3snm1/2} The $\SigOnly$ rule is Condorcet-consistent, monotone, \newline $\kSNMx{2}{1/3}$ and $\kSNMx{3}{1/2}$, and this is tight.~\footnote{The best lower bound for this problem, due to~\cite{RSEB}, is $2/5$.}
\end{restatable}

The $\SigOnly$ rule, while explicitly describable, arises from an approach similar to the LP-based rule of~\cite{LPmanipulability}. We identify the tournaments which are close to having a Condorcet-winner as those where teams have more incentives to manipulate outcomes. Given such a close-to-Condorcet tournament $T$, our rule deems a small number of teams as \emph{significant} and distributes most of the probability mass on these teams according to additional properties of the tournament itself (and the rest uniformly across the remaining teams).

Finally, our last contribution introduces a new way of designing Condorcet-con-sistent and \emph{asymptotically} optimal tournament rules. If one had substantial computational power, one could compute a top-cycle consistent,\footnote{Top-cycle consistency is stronger than Condorcet-consistency. We defer its formal definition but informally, the top-cycle is the smallest non-empty set $S$ of teams such that no team in $S$ loses to a team outside of $S$. A top-cycle consistent rule would only pick teams from the top-cycle.} optimal rule for fixed values of $n, k$. How could we use such a rule $r_n$ to construct a rule $r_{n'}$ that works for $n' > n$? First pad the tournament with dummy teams that lose to all real teams until the number of teams $n' := n\cdot M$ is a multiple of $n$. Partition the teams into $n$ groups of equal size. Within each group, pick a team from the top-cycle uniformly at random as a finalist. The number of finalists will be exactly $n$. Finally, run $r_n$ on the $n$ finalists and declare its winner as the overall winner. We prove that this simple idea suffices to transform top-cycle consistent, $\kSNMx{k}{\alpha}$ rules for $n$ teams to CC, $\kSNMx{k}{\alpha'}$ rules for $n' > n$ teams where $\alpha'$ is close to $\alpha$.  

\begin{restatable}{theorem}{thmreduction} \label{thm:inc_num_of_teams_arbitrarily_top-cycle}
If there exists a top-cycle consistent, and $\kSNMx{k}{\alpha}$ rule $r$ for $n$ teams, then there exists a top-cycle consistent and $\kSNMx{k}{\alpha'}$ rule $r'$ for $n' > n$ teams where
\[\alpha' \leq \alpha \Big( 1 - \frac{(k-1)^2}{n} \Big) + \frac{(k-1)^2}{n}.\]
\end{restatable}

Theorem~\ref{thm:inc_num_of_teams_arbitrarily_top-cycle} can be thought of as reducing the problem of finding approximately optimal tournament rules for large $n$ to the same problem for small $n$. As an example, if we could verify the existence of a top-cycle consistent $\kSNMx{3}{2/5}$ rule for $n = 25$, then this would imply top-cycle consistent $\kSNMx{3}{\alpha}$ rules for $n \geq 25$ and $\alpha < 1/2$, directly improving on Theorem~\ref{thm:main3snm1/2}.


\subsection{Related Work}
\label{sec:related}

Most of the related work has been mentioned already. There exist two other results that are directly related to our problem. Whereas in this paper (and all previously mentioned papers) we evaluate the manipulability of tournament rules based on their worst-case performance, the work of~\cite{probabilisticdeathmatch} instead studies this question under the lens of average-case analysis. More recently, the work of~\cite{tournamentprizes} expanded the model to include prize vectors $v = (v_1, \dots, v_n)$. Tournament rules with prizes output a complete, linear ranking over the teams where the $i$-th team earns reward $v_i$, rather than giving reward $1$ to the winner and 0 to everyone else.

Another related line of work involves the Tournament Fixing Problem, where the organizer of the tournament is colluding with a team in order to produce a seeding that selects them as the winner (see, e.g.~\cite{Bartholdi92,VuAS09,Williams10,StantonW11,KimW15,KimSW16}). The central questions here are computational (i.e., can the organizer efficiently decide if there is a winning bracket for their favorite team) and structural (i.e., under what conditions does there exist a bracket that selects the organizer's favorite team, see, e.g.~\cite{Maurer80}). Due to its connections with voting theory and social choice, there is a long history of analyzing properties of particular tournament rules (\cite{Fishburn77,Copeland51,Gibbard73,Miller80,Moulin86,Dutta88,Csato17}, to name a few). We refer the reader to the survey by~\cite{surveyWarut} on recent developments in tournaments and computational social choice, or other books on computational social choice (\cite{BCELP16,Laslier1997}).  
\section{Notation} 
\label{sec:notation}

In this section,we introduce key concepts to contextualize our results. Recall a tournament graph $T$ is a complete, directed graph $G = ([n], E)$ on $n$ labelled vertices. We refer to a tournament graph's vertices as teams (and the number of teams as the size of the tournament), its undirected edges as matches and its directed edges as outcomes (where if $(i, j) \in E,$ we say $i$ beats $j$ in $T$). For a fixed team $i$, tournament $T$ we let $\delta^{+}(i, T) = \{ j | j \in [n], (i, j) \in E) \}$ be the set of teams that $i$ beats under $T$, and let $\delta^{-}(i, T) = \{ j | j \in [n], (j, i) \in E) \}$ be the set of teams that $i$ loses to under $T$. Let $\mathcal{T}_n$ be the set of all tournaments on $n$ teams. Recall a tournament rule is a mapping $r: \mathcal{T}_n \rightarrow \Delta^n$. That is, for every tournament $T$, $r(T)$ denotes the distribution over teams according to which the organizer will select a winner. We use notation $r_i(T)$ to denote team $i$'s probability of being selected by $r$ as the winner under tournament $T$. We use the shorthand notation $r_S(T) := \sum_{i \in S} r_i(T)$ to denote the probability that a team in $S$ is selected by $r$ in tournament $T$. The next definitions formalize Properties~\ref{prop1},~\ref{prop2},~\ref{prop3} and ~\ref{prop3r}. 

\begin{definition}
\label{def:condorcet} Team $i$ is the \emph{Condorcet-winner} of tournament $T$ if $i$ beats every other team under $T$. A rule $r$ is \emph{Condorcet-consistent} if $r_i(T) = 1$ when $i$ is $T$'s Condorcet-winner.  
\end{definition}

\begin{definition}
\label{def:monotone} A tournament rule $r$ is \emph{monotone} if for all teams $i$ and all tournaments $T, T'$ where all matches not involving team $i$ are identical and $\delta^{+}(i, T) \supseteq \delta^{+}(i, T')$, it holds that $r_i(T) \geq r_i(T')$.   
\end{definition}

A tournament rule is monotone if it is not in a team's best interest to unilaterally lose matches it would otherwise win. We present the ways in which manipulations are modelled. 

\begin{definition}
\label{def:sAdjT} We say tournaments $T, T'$ are \emph{$S$-adjacent} if the only outcomes where $T, T'$ differ on are those matches that involve two teams in $S$.  
\end{definition}

 If $T, T'$ are $S$-adjacent, outcomes involving at least one team outside of $S$ are identical. Motivated by the results from~\cite{nonmanipulable_randomized_tournament_selections}, the following relaxation was introduced by~\cite{RSEB}. 


\begin{definition}
\label{def:ksnmalpha} A tournament rule is \emph{$k$ strongly non-manipulable at probability $\alpha$} \newline ($\kSNMx{k}{\alpha}$) if for all $S \subseteq [n]$ of size at most $k$, for all tournaments $T, T'$ that are $S$-adjacent we have $r_S(T') \leq r_S(T) + \alpha$. For $\alpha = 0$, we simply say the rule is \emph{$k$ strongly non-manipulable} ($\kSNM{k}$).
\end{definition}
\section{Analysis for the \RdSEB{d} Rule}
\label{sec:d-ary}

In this section we study the manipulability of the Randomized $d$-ary Single-Elimination Bracket (\RdSEB{d}) rule, a generalization of the Randomized Single-Elimination Bracket ($\RSEB$) rule from~\cite{RSEB}, against collusions of size $k \leq d$. 



\begin{definition}
\label{def:subT} Given a tournament $T$, a \emph{sub-tournament} on $S$ is the sub-graph induced by $T$ on vertex set $S$.  
\end{definition}

We are now ready to formally define the $\RdSEB{d}$ rule. 

\begin{definition}
\label{def:rdseb}
The \emph{Randomized $d$-ary Single-Elimination Bracket} rule operates as follows. Add dummy teams~\footnote{Dummy teams are teams that lose to all non-dummy teams. The outcome of a match between two dummy teams is arbitrary.} until the number of teams $n = d^{\lceil \log_d (n) \rceil}$ is a power of $d$. Randomly place teams at the leaf nodes of a complete $d$-ary tree of height $\lceil \log_d(n) \rceil$. Recursively label a parent node with the label of the Condorcet-winner of the sub-tournament induced by the labels of its children, if there is one. Otherwise, choose one of its children uniformly at random and use that label instead. The winner of the tournament is the team whose label appears at the root of the tree. 
\end{definition}

In terms of Definition~\ref{def:rdseb}, $\RSEB$ is the rule that results from setting $d=2$. The family of \RdSEB{d} rules operates in the same way as the $\RSEB$ rule except that if there is no Condorcet-winner in the sub-tournament induced at a node, the rule advances a team uniformly at random. This choice is motivated by the following simple observation when $d = 3$. There are only two non-isomorphic sub-tournament graphs on three teams: one where there is a Condorcet-winner and one where the teams beat each other cyclically. In the former case, it is obvious which team to advance. In the latter case, we argue choosing a team uniformly at random is reasonable.


\thmRdSEB*

The main idea is that if some (at least two) colluding teams meet only in the final round, then they can increase the joint probability that one of them will be winner by at most $(1 - 2/d)$ (which happens when the colluding teams create a~Condorcet winner). Obviously, if the colluding teams do not meet in a~bracket at all or there exists a Condorcet winner outside the colluding teams, then they cannot increase the chance to win the tournament. In the remaining cases, we simply assume that they can increase the chance to win the tournament by~1.

From Theorem~\ref{thm:RdSEB}, given a fixed $k$, the \RdSEB{$k$} rule is monotone, CC and its manipulability is bounded away from $1$ for all $n$. This is the first explicit family of rules to exhibit this property (since the LP-based rule of~\cite{LPmanipulability} is neither monotone nor explicit). We suspect the bound from Theorem~\ref{thm:RdSEB} is not tight. A finer argument like the one in \cite{RSEB} might yield a better analysis. 

\begin{proof}
        %
        %
        %
        %
              It is more convenient to consider the following equivalent variation of the \newline $\RdSEB{d}$ rule.
        %
        %
        %
        %
        Instead of choosing one of the children of~$A$ uniformly at random, chose a~number~$j_A$ from $[d]$ uniformly at random.
        If there is no Condorcet-winner in the sub-tournament, label~$A$ by the label of $j_A$-th child of~$A$.
        In both cases, mark the node~$A$ with~$j_A$ (even if there is a~Condorcet-winner in the sub-tournament).\footnote{Every inner node has a~label and a~mark. Note that they can be equal but they have different meaning.}
        For the sake of the proof, let us denote~$r^d$ as this equivalent variation of the $\RdSEB{d}$ rule.
        Observe that the labels of inner nodes of the complete $d$-ary tree can be deduced from the labels of leaves and marks of inner nodes: run~$r^d$ but all random choices are made accordingly to these labels and marks.

        For any non negative number~$t$, let~$D_t$ be a~reserved set containing~$t$ dummy teams.
        Moreover, let~$T$ be a~tournament on~$N'$ (disjoint from any~$D_t$) of~$n'$ teams, $h := \lceil \log_d(n') \rceil$, and $n := d^h$. 
        We assume that the rule~$r^d$ initially adds~$n-n'$ dummy teams from $D_{n-n'}$  into~$T$.\footnote{Hence every tournament on~$n'$ teams is extended by the same set of dummy teams.}
        Given $N := N' \cup D_{n-n'}$, a~\emph{$d$-bracket} $G(\pi, m)$ for~$N$ is a~complete $d$-ary tree~$G$ of height~$h$ endowed with a~pair $(\pi, m)$, where
        \begin{itemize}
            \item $\pi$ is a~bijection from leaves of~$G$ to~$N$ (i.e., labels of leaves),
            \item $m$ is a~mapping from inner nodes of~$G$ to~$[d]$ (i.e., marks of inner nodes).
        \end{itemize}
        Let~$\mathcal{B}_N$ be the set of all $d$-brackets for~$N$. Now we precisely describe how the labels of nodes of~$G$ can be deduced from~$\pi$ and~$m$.
        Given a~$d$-bracket $G(\pi, m)$ for~$N$, the \emph{outcome} of $G(\pi,m)$ under $T$\footnote{The outcome is well-defined only if the teams in~$T$ are the same as the set~$N$.} is a~labeling~$\omega_T$ of nodes of~$G$ such that $\omega_T(A) = \pi(A)$, for every leaf~$A$, and if~$A$ is a~node with children $A_1, \dots, A_d$, then
        \[ \omega_T(A) := \begin{cases} 
            x & \text{if } x \text{ is the Condorcet-winner in } \\
            & \text{the sub-tournament induced by} \\
              & \{\omega_T(A_1), \dots, \omega_T(A_d)\} \text{ under } T, \\
            \omega_T(A_{m(A)}) & \text{otherwise}.
        \end{cases}
        \] 
        Observe that~$\omega_T$ is a~one-to-one correspondence between the set of all outcomes of $d$-ary brackets on~$N$ under~$T$ and the set of all runs of~$r^d$ on~$T$.

        Fix a~$d$-bracket $G(\pi, m)$ for~$N$.
        The \emph{winner} of $G(\pi, m)$ under~$T$ is $\omega_T(R)$, where~$R$ is the root of~$G$.
        Given a~team $x \in N$, $G(\pi, m)$ is \emph{winning for}~$x$ under~$T$ if~$x$ is the winner of $G(\pi, m)$ under~$T$.
        Denote by $\mathcal{B}_{N,T}(x) \subseteq \mathcal{B}_N$ the set of all winning $d$-brackets for~$x$ under~$T$.
        The motivation behind this reframing of $\RdSEB{d}$ is to simply the argument of the proof. Similar to the original argument of \cite{RSEB}, we will bound the manipulability of $\RdSEB{d}$ by directly counting the number of brackets where colluding teams could gain and compare it to the total number of brackets. We have \[r^d_x(T) = |\mathcal{B}_{N,T}(x)|/|\mathcal{B}_N|.\]
        Observe that $|\mathcal{B}_N| = n! \cdot d^{\ell}$, where~$\ell := \ell(d,h)$ is the number of inner nodes of a~complete $d$-ary tree of height~$h$.

        First, we prove that~$r^d$ is monotone. Take an arbitrary team $x \in N'$.
        It sufficient to show that $r^d_x(T) \geq r^d_x(T')$ for every $\{x,y\}$-adjacent tournaments~$T'$ of~$T$ such that~$x$ beats~$y$ under~$T$.
        Let $G(\pi, m)$ be a~$d$-bracket for $N$. Observe that if $G(\pi, m)$ is winning for~$x$ under~$T'$, then $G(\pi, m)$ is also a~winning for~$x$ under~$T$.
        Hence $\mathcal{B}_{N,T'}(x) \subseteq \mathcal{B}_{N,T}(x)$, and so $r^d_x(T') \leq r^d_x(T)$ as required.

        Second, we prove that~$r^d$ is Condorcet-consistent. Suppose that~$x\in N'$ is the Condorcet-winner in~$T$.
        For every $d$-bracket $G(\pi, \ell)$ for~$N$, consider the unique path~$P$ from the leaf in~$G$ labeled by~$x$ to the root of~$G$.
        Observe that every node of~$P$ is labeled by~$x$ by the outcome of $G(\pi, m)$ under~$T$. In particular, the root of~$G$ is labeled by~$x$.
        It follows that every $d$-bracket $G(\pi, \ell)$ is winning for~$x$ under~$T$. Hence $\mathcal{B}_{N,T}(x) = {B}_{N}$, and so $r^d_x(T) = 1$ as required.

        Lastly, we prove that~$r^d$ is $\kSNMx{k}{\alpha_{d,k}}$ for some~$\alpha_{d,k}$ (to be determined).
        Suppose that $S = \{s_1, s_2, \dots, s_k\} \subseteq N'$ is a~subset of colluding teams. 
        For any $S$-adjacent tournaments~$T'$ of~$T$, we show that $r^d_S(T') - r^d_S(T) \leq \alpha$.
        Recall that $\mathcal{B}_{N,T}(x)$ is the set of all winning $d$-brackets for~$x$ under~$T$. Moreover, define $\mathcal{B}_{N,T}(S) := \bigcup_{s \in S} \mathcal{B}_{N,T}(S)$. In this notation, we can write
        \[ r_S(T') - r_S(T) = \frac{|\mathcal{B}_{N, T'}(S)|}{|\mathcal{B}_{N}|} - \frac{|\mathcal{B}_{N, T}(S)|}{|\mathcal{B}_{N}|}. \]
        
        We upper bound this expression using the idea introduced in~\cite{RSEB}.
        For that, let us denote by $\mathcal{B}_{N}^{+}(S)$ the set of all $d$-brackets~$G(\pi, m)$ for~$N$ such that the least common ancestor of any leaves $A$ and $B$ with $\pi(A), \pi(B) \in S$ is the root of~$B(\pi, m)$.
        In other words, $\mathcal{B}_{N}^{+}(S)$ is the set of all $d$-brackets for~$N$  such that the colluding teams can meet possibly only in the final round. Set $\mathcal{B}_{N}^{-}(S) := \mathcal{B}_{N} \setminus \mathcal{B}_{N}^{+}(S)$.
        Moreover, for a~tournament~$T$, let $\mathcal{B}_{N,T}^{+}(S) := \mathcal{B}_{N}^{+}(S) \cap \mathcal{B}_{N,T}(S)$ and $\mathcal{B}_{N,T}^{-}(S) := \mathcal{B}_{N}^{-}(S) \cap \mathcal{B}_{N,T}(S)$.
        Then 
        \begin{align*}
            r_S(T') - r_S(T) &= \frac{|\mathcal{B}_{N, T'}(S)|}{|\mathcal{B}_{N}|} - \frac{|\mathcal{B}_{N, T}(S)|}{|\mathcal{B}_{N}|} \\
                = \frac{|\mathcal{B}_{N, T'}^{+}(S)| + |\mathcal{B}_{N, T'}^{-} (S)|}{|\mathcal{B}_{N}|} &- \frac{|\mathcal{B}_{N, T}^{+} (S)| + |\mathcal{B}_{N, T}^{-} (S)|}{|\mathcal{B}_{N}|} \\
                = \frac{|\mathcal{B}_{N, T'}^{+} (S)| - |\mathcal{B}_{N, T}^{+} (S)|}{|\mathcal{B}_{N}|} &+ \frac{|\mathcal{B}_{N, T'}^{-} (S)| - |\mathcal{B}_{N, T}^{-} (S)|}{|\mathcal{B}_{N}|} \\
                \leq \frac{|\mathcal{B}_{N, T'}^{+} (S)| - |\mathcal{B}_{N, T}^{+} (S)|}{|\mathcal{B}_{N}|} &+ \frac{|\mathcal{B}_{N}^{-}(S)|}{|\mathcal{B}_{N}|}.
        \end{align*}
     
        We upper bound the first term in the last expression. 
        Let $G := G(\pi,m)$ be a~$d$-bracket for~$N$ in $\mathcal{B}_{N, T'}^{+}(S)$.
        Let $\textbf{x} := (x_1, \dots, x_d)$ be the $d$-tuple of finalists in~$G$ under~$T'$.
        More precisely, let~$R$ be the root of~$G$ with children $A_1, \dots, A_d$. 
        Then $d$-tuple of finalists of~$G$ under~$T$ is $(\omega_{T'}(A_1), \dots, \omega_{T'}(A_d))$.
        The crucial observation is that also~$\textbf{x}$ is $d$-tuple of finalist of~$G$ under~$T$.
        We say that there is a~Condorcet-winner in~$\textbf{x}$ under~$T'$ if there is a~Condorcet-winner in the sub-tournament induced by~$x$ under~$T'$.
        Notice that:
        \begin{enumerate}[(i)]
            \item If $|S \cap \{x_1, x_2, \dots, x_d\}| \leq 1$, then $G \in \mathcal{B}_{N, T}^{+}(S)$.
            \item If $|S \cap \{x_1, x_2, \dots, x_d\}| \geq 2$ and there is no Condorcet-winner in~$\textbf{x}$ under~$T'$, then $G \in \mathcal{B}_{N, T}^{+}(S)$.
            \item If $|S \cap \{x_1, x_2, \dots, x_d\}| \geq 2$ and there is a~Condorcet-winner $s_i \in S$ in~\textbf{x} under~$T'$, then $G \notin \mathcal{B}_{N, T}^{+}(S)$ only if there is no Condorcet-winner $s_j \in S$ in~\textbf{x} under~$T$ and the mark of the root~$R$ of~$G$ is pointing to a~team outside of~$S$ (i.e., $\omega_T(A_{m(R)}) \notin S$).
        \end{enumerate}
        A~$d$-bracket is of type~(i) if it satisfies the statement~(i). Analogously, for types~(ii) and~(iii).
        It follows that, for every $d$-bracket in~$\mathcal{B}_{N, T'}^{+}(S)$ of type~(i) or~(ii), there exists at least one $d$-bracket in~$\mathcal{B}_{N, T}^{+}(S)$ of the same type.
        Moreover, we claim that, for every $d$~$d$-bracket in~$\mathcal{B}_{N, T'}^{+}(S)$ of type~(iii), there exist at least two $d$-brackets in~$\mathcal{B}_{N, T}^{+}(S)$ of type~(iii).
        Indeed, take a~bracket $G(\pi, m) \in B_{N,T'}^{+}(S)$ of type~(iii). Then $G(\pi, m') \in B_{N,T'}^{+}(S)$ is of type~(iii), where only the mark of the root of~$G$ can be changed (there are~$d$ ways how to change it).
        On the other hand, let $i \neq j$ be two indices such that $x_i, x_j \in S$.
        Then $G(\pi, m_1), (\pi, m_2) \in B_{N,T}^{+}(S)$ are of type~(iii), where~$m_i$ and~$m_j$ are $m$ but the mark of the root of~$G$ is changed to~$i$ and~$j$, respectively.
        
        If we denote by~$p$ be the number of $d$-brackets from~$\mathcal{B}_{N, T'}^{+}(S)$ of type~(iii), then
        \begin{align*}
        &|\mathcal{B}_{N, T'}^{+} (S)| - |\mathcal{B}_{N, T}^{+} (S)| \leq p \cdot \Big( 1 - \frac{2}{d} \Big) \\
        &\leq |\mathcal{B}_{N, T'}^{+}(S)| \cdot \Big( 1 - \frac{2}{d} \Big) \leq |\mathcal{B}_{N}^{+}(S)| \cdot \Big(1 - \frac{2}{d} \Big).
        \end{align*}
        Observe that \[ |\mathcal{B}_{N}^{+}(S)| = (d)_k \cdot \Big( \frac{n}{d} \Big)^k \cdot (n-k)! \cdot d^{\ell} \]
        and hence \[ |\mathcal{B}_{N}^{-}(S)| = n! \cdot d^{\ell} - (d)_k \cdot \Big( \frac{n}{d} \Big)^k \cdot (n-k)! \cdot d^{\ell}. \]
        Therefore, 
        \begin{align*}
            r_S(T') - r_S(T) &\leq \frac{|\mathcal{B}_{N, T'}^{+} (S)| - |\mathcal{B}_{N, T}^{+} (S)|}{|\mathcal{B}_{N}|} + \frac{|\mathcal{B}_{N}^{-}(S)|}{|\mathcal{B}_{N}|} \\
                &\leq \frac{\Big( (d)_k \cdot \big( \frac{n}{d} \big)^k \cdot (n-k)! \cdot d^{\ell} \Big) \cdot \Big( 1 - \frac{2}{d} \Big)}{n! \cdot d^{\ell}} \\
                &+ \frac{n! \cdot d^{\ell} - (d)_k \cdot \big( \frac{n}{d} \big)^k \cdot (n-k)! \cdot d^{\ell}}{n! \cdot d^{\ell}} \\
                &= 1 - \frac{1}{n!} \cdot \Big(\frac{2}{d} \cdot (d)_k \cdot \big( \frac{n}{d} \big)^k \cdot (n-k)! \Big) \\
                &= 1 - \frac{n^k \cdot (n-k)}{n!} \cdot \frac{2 \cdot (d)_k}{d^{k+1}} \\
                &\leq 1 - \frac{2 \cdot (d)_k}{d^{k+1}}.
        \end{align*}
    \qed
    \end{proof}

We conjecture that the manipulability of $\RdSEB{d}$ is bounded away from $1/2$ for all $k \leq d$. 

\begin{conjecture}[See Table~\ref{table__RSdEB}]
For all $3\leq k \leq d$, $\alpha_{d,k} \geq 227/420$. 
\end{conjecture}

    \begin{table}[t]
        \centering
        \begin{tabular}{|c|c|c|c|c|c|c|c|c|}
        \hline
           \diagbox{$d$}{$k$} & 3 & 4 & 5 & 6 & 7 \\ \hline
           3 & 0.8519 & - & - & - & -  \\ \hline
           4 & 0.8125 & 0.9531 & - & - & - \\ \hline
           5 & 0.808  & 0.9232 & 0.9846 & - & - \\ \hline
           6 & 0.8148 & 0.9074 & 0.9691 & 0.9949 & -\\ \hline
           7 & 0.8250 & 0.9 & 0.9572 & 0.9878 & 0.9983 \\ \hline
        \end{tabular}
        \vspace{0.3cm}
        \caption{Evaluation of $\alpha_{d,k}$ for small values of~$d$ and~$k$ rounded up to 4 decimals.}
        \label{table__RSdEB}
    \end{table}


    %

\section{Analysis for the \SigOnly Rule}
\label{sec:sigOnly}

In this section we formalize the \SigOnly~rule and outline the proof of Theorem~\ref{thm:main3snm1/2}. The motivation behind the $\SigOnly$ rule is captured by the following simple observation. Suppose we are trying to design a rule which is CC and \kSNMx{$k$}{$\alpha$} for some fixed $k$. Take any tournament graph $T$. If there is a Condorcet-winner, then the rule is fixed and must declare that team the winner. If $T$ is \emph{far} from having a Condorcet-winner, meaning that the smallest set of teams who could collude and produce a Condorcet-winner among them is larger than $k$, then pick a team uniformly at random. If $T$ is only a small number of manipulations away from having a Condorcet-winner, then in order to satisfy the $\SNM$ constraint (approximately) we must allocate all (resp. most) of the probability mass to the teams who could produce a Condorcet-winner. However, we must be careful in order not to incentivize teams in tournaments that are far from having a Condorcet-winner to manipulate into those which are close to having a Condorcet-winner. 

The previous paragraph captures the spirit of the \SigOnly~rule. We first partition the set of all tournament graphs $\mathcal{T}_n$ into three groups: those with a Condorcet-winner (\emph{Condorcet} tournaments), those where (maybe multiple) sets of at most $k$ teams can produce a Condorcet-winner (\emph{near-Condorcet} tournaments) and those where no set of $k$ teams can produce a Condorcet-winner (\emph{far-Condorcet} tournaments). The rule is fixed for the first part of the partition, and will select a winner uniformly at random on the last part of the partition. The middle part of the partition is further partitioned into four categories depending on the exact number and size of the groups that could produce a Condorcet-winner. For each of the parts in this sub-partition, we propose a way to distribute the probability mass. We must balance two different sets of incentives here. On the one side, we must give sufficient probability mass to the groups that can produce a Condorcet-winner. On the other, we can't give them \emph{too} much mass, since otherwise teams in far-Condorcet tournaments might be substantially incentivized into manipulating the tournament into a near-Condorcet one, where some of the colluding members are also in groups that can now produce Condorcet-winners. We formalize the above with the following definitions.

\begin{definition}
\label{def:winninggroup} A tournament $T$ is said to be \emph{near-Condorcet} for $k$ if there is no \newline Condorcet-winner in $T$ but there exists at least one team $i$ with $|\delta^{-}(i, T)| \leq k-1$. Call the set of teams $MW(i, T) := \{i\} \cup \delta^{-}(i, T)$ a \emph{minimal winning group} (MW group). We call team $i$ the \emph{leader} of $MW(i,T)$ and any team in $j \in MW(i,T)$ \emph{significant}. If $|MW(i,T)| = 2, 3$ we call it an MW pair or an MW triple, respectively. 
\end{definition}

We first prove simple structural properties of near-Condorcet tournaments. 

 \begin{lemma}\label{lemma:uniqueleader}
        Every minimal winning group has exactly one leader.
    \end{lemma}
    
    \begin{proof}
        Suppose there exists some other leader $j$ of $MW(i,T)$. By definition, $j$ must lose to $i$ and vice versa, a contradiction. 
        \qed
    \end{proof}
    
    Notice that even though every MW group has a different unique leader, a leader of one group can be a member in a different group. Moreover, it can happen that one MW group is a subset of another MW group.

    \begin{lemma}\label{lemma:groupintersection}
         Let $T$ be a near-Condorcet tournament, and $MW(i,T)$ and $MW(j, T)$ be distinct MW groups in $T$ with leaders $i$ and $j$, respectively. Then either $i$ is in $MW(j,T)$, or $j$ is in $MW(i,T)$. In particular, $MW(i,T)$ and $MW(j,T)$ must have a non-empty intersection.
    \end{lemma}
    
    \begin{proof}
    Suppose that both $i$ and $j$ in $MW(i,T) \cap MW(j,T)$, or they are both outside. Then, by the definitions of $MW(i,T)$, $MW(j,T)$, they should beat each other, a contradiction. Thus, exactly one of them must be in the intersection.  
    \qed
    \end{proof}
    
 We now prove the main lemma about the structure of near-Condorcet tournaments for $k=3$, showing that there are not too many significant teams and leaders. 
    
    \begin{lemma}\label{lemma:sizeofpart}
        If $k=3$ and $T$ is a near-Condorcet tournament, then the number of significant teams in $T$ is at most $6$, and the number of teams in the union of MW pairs is at most 3. Furthermore, there cannot be more than 3 MW pairs. If there is exactly one MW pair, then the maximal number of significant teams is $5$, if there are exactly two $MW$ pairs, then the maximal number of significant teams is $4$, and if there are exactly three pairs, then the maximal number of significant teams is $3$.
    \end{lemma}

\begin{proof}
        Assume there are $p$ MW pairs and $t$ MW triples in $T$. By Lemmas \ref{lemma:uniqueleader} and \ref{lemma:groupintersection}, we know that each such group has a unique leader, and these leaders are pairwise distinct. Furthermore, each leader of an MW pair loses exactly one match, and each leader of an MW triple loses exactly two matches. Therefore, all leaders together lose exactly $p+2t$ matches. On the other hand, there are exactly $\binom{p+t}{2}$ matches between leaders, and in each of these matches, one leader loses. Hence,
        \[\binom{p+t}{2} \le p+2t.\]
        
        We can equivalently rewrite is as $p^2+t^2+2pt-3p-5t \le 0$. If $p\ge4$ we get that $4+3t+t^2\le 0$, which does not have a solution for a non-negative integer $t$. Hence, there are either 3, 2, 1, or no MW pairs. Furthermore, significant teams are either leaders of some groups or they beat some leaders. Hence, there are at most the number of leaders plus the number of matches lost by leaders to some non-leaders many of them. In other words, there at most $(p+t) + p + 2t - \binom{p+t}{2}$ significant teams. If there are three pairs, the maximum of this expression for non-negative integer $t$ is 3. Similarly, if $p=2$ the maximum is 4, if $p=1$ the maximum is 5, and if $p=0$ the maximum is 6. Moreover, there cannot be more than 3 teams in the union of all MW pairs. Otherwise, there would be at least 3 MW pairs (there cannot be only two because they intersect), but we already showed that if we have three or more MW pairs, there can be at most 3 teams in the union of all MW groups. 

        
    \qed
    \end{proof}
    
    This lemma implies that in every near-Condorcet tournament at most 6 teams can be directly part of some manipulating group creating a Condorcet-winner. Hence, we can directly design a rule that is $3$-SNM-$\frac{2}{3}$. It is sufficient to assign probability 1 to Condorcet-winners, probability $\frac{1}{6}$ to significant teams in near-Condorcet tournaments, and distribute the remaining probabilities in all tournaments uniformly between the remaining teams. We will not prove this formally since we design a better rule, but we believe this intuition is useful in understanding our construction. We now formalize the partition of tournament graphs $\mathcal{T}_n$.  
	\begin{itemize}
	    \item Let $\mathcal{CC}_n \subset \mathcal{T}_n$ be the set of tournaments with a Condorcet-winner.
	    \item Let $\mathcal{FC}_n \subset \mathcal{T}_n$ be the set of far-Condorcet tournaments. 
	    \item Let $i\mathcal{-NCP}_n \subset \mathcal{T}_n$ be the set of near-Condorcet tournaments with exactly $i$ MW pairs for $i=0,1,2,3$. 
	\end{itemize}
	
    Now we are ready to define our main rule.
	
	\begin{definition}[\SigOnly~Tournament Rule]
	    For $n \geq 6$ teams, the \SigOnly~\emph{tournament rule} does the following. 
	    
	    \begin{enumerate}
        
	        \item If $T \in \mathcal{CC}_n$ the Condorcet-winner gets 1, and the remaining teams get zero.
	        
	        \item If $T \in \mathcal{FC}_n $, pick a winner uniformly at random.
         
	        \item If $T \in 3\mathcal{-NCP}_n \cup 2\mathcal{-NCP}_n$, teams in MW pairs get \(\frac{1}{3}\), and the remaining teams get zero.
         \setlength{\itemsep}{0.2ex}
         
	      \setlength{\itemsep}{0ex}  
	        \item If $T \in 1\mathcal{-NCP}_n$, teams in the only MW pair get $\frac{1}{3}$, teams in MW triples that are not in any MW pair get $\frac{1}{9}$, and the remaining teams get the remaining probability mass evenly distributed between them.
	        \item If $T \in 0\mathcal{-NCP}_n$, teams in MW triples get $\frac{1}{6}$, and the reaming teams get the remaining probability mass evenly distributed between them.
	    \end{enumerate}
	\end{definition}
	
	We restate the main result of this section.
 
    \thmsnm*
	
	The fact that $\SigOnly$~is Condorcet-consistent  follows directly from the definition. We defer the rest of this section, including the proof of Theorem~\ref{thm:main3snm1/2} as well as the instance that witnesses its tightness to Appendix~\ref{app:3snm}. 

	
\section{Almost Optimal Rules for Large $n$ from Optimal Rules for Small $n$}
\label{sec:reduction}

For fixed values of $n, k$, \cite{RSEB} present an LP that can compute the minimally manipulable, Condorcet-consistent and monotone rule. Unfortunately, solving this LP is highly intractable as the number of variables and constraints grows exponentially in the number of teams. For small values of $n, k$, and with sufficient computational power, one could compute such solutions, in the hope that one could use that rule in order to construct approximately optimal ones for larger $n$. Consider the following simple procedure to scale a $\kSNMx{k}{\alpha}$ rule $r_n$ for $n$ teams to a $\kSNMx{k}{\alpha'}$ rule $r_{n'}$ for $n' > n$. First, increase $n'$ to the nearest multiple of $n$, $n': = nM$ of $n$. Partition the teams into $n$ groups of equal size. Within each group, compute the top-cycle on the induced sub-tournament and select a team from it uniformly at random as a finalist. This will reduce the number of teams to exactly $n$ finalists. Then run $r_n$ on the $n$ finalists and declare that rule's winner as the overall winner. The result of this section, restated below, bounds the value of $\alpha'$ for $r_{n'}$. 

\thmreduction*

We now outline the proof of Theorem~\ref{thm:inc_num_of_teams_arbitrarily_top-cycle} and defer its proof to Appendix~\ref{app:reduction}. 
    It is easy to show that $r'$ is top-cycle consistent.
    The derivation of the upper bound on~$\alpha'$ is more complicated but it is based on a~fairly simple idea. Let~$S$ be a~set of colluding teams in a~tournament~$T$ on a set of $n'$ teams.
    If a~permutation~$\pi$ on a~set~of $n'$ teams is chosen uniformly at random, then either every group in the partition contains at most one team from~$S$ or there exists a group containing at least two teams from~$S$.
    In the former case, we use the fact that fixing matches inside~$S$ does not change a team's chances of surviving the group. Then we use the assumption that~$r$ is $k$-SNM-$\alpha$ to conclude that~$S$ can increase the probability to win the tournament~$T$ by at most~$\alpha$ by fixing matches inside~$S$. 
    In the latter case, we simply assume that they can increase the probability by~1.
    The latter case happens with probability $(k-1)^2/n$, which finishes the proof. An explicit consequence of this result is the following. If there exists a top-cycle consistent \kSNMx{3}{2/5} rule for $n = 25$ teams (which would be the best possible as per \cite{RSEB}), then there exists a \kSNMx{3}{$\alpha$} rule for all $n \geq 25$ for some $\alpha < 1/2$. 

\section{Conclusion and Future Directions}
\label{sec:conclusion}

This paper extends our knowledge of non-manipulable tournament rules in several ways. First, we generalize \RSEB~ from~\cite{RSEB} into $\RdSEB{d}$. This rule, at every node, picks a Condorcet-winner if one exists among its children and otherwise chooses a team uniformly at random. We show that for $k \leq d$ this rule is $\kSNMx{k}{\alpha_{d,k}}$ for some $\alpha_{d, k}$ bounded away from $1$, providing the first explicit family of rules that are bounded away from $1$ for any $n$. We suspect that a more careful analysis for small values of $d, k$ might yield better rules but conjecture, however, that $\alpha_{d, k} > 1/2$ for all $k \leq d$. 

Then we present a novel rule, the \SigOnly Rule, which is Condorcet-consistent, monotone, $\kSNMx{3}{1/2}$ (improving on concurrent work of~\cite{deathmatch}) and \newline $\kSNMx{2}{1/3}$ (which is best possible). The rule identifies a small set of teams as significant, awards them substantial probability mass and distributes it uniformly among non-significant teams. The motivation for the rule is that tournaments with a Condorcet-winner and tournaments which are far from having a Condorcet-winner are easy to resolve. We find a way to resolve the intermediate tournaments in a way to avoid substantial gains from manipulation.

Finally we propose a way of reducing the problem of finding good rules for large values of $n$ to the problem of finding good rules for small values of $n$. We show that given a rule for a small number of teams, there is a simple way to extend it to more teams while keeping the manipulability guarantees relatively close. Our result implies that if there exists a $\kSNMx{3}{2/5}$ just for $n = 25$, then there exist $\kSNMx{3}{\alpha}$ rules for $n \geq 25$ and $\alpha < 1/2$, which would directly improve on the state of the art results. 

Overall, our results introduce new ways to overcome obstacles in designing and analyzing approximately optimal tournament rules for collusions of size 3 or more.

\bibliographystyle{acm}
\bibliography{references.bib}
\newpage

\appendix
\section{Missing Proofs from Section~\ref{sec:sigOnly}}
\label{app:3snm}

In this Appendix, we include the missing proofs from Section \ref{sec:sigOnly}. We first start with some simple properties about the $\SigOnly$ rule. 

	\begin{lemma}\label{lem:properlydefined}
	    The \SigOnly~rule is a properly defined tournament rule. 
	\end{lemma}

\begin{proof}
	    $\SigOnly$ assigns probabilities to all teams in all tournaments. These probabilities are non-negative. Hence, we just need to check that they sum up to 1 in every tournament.
	    For tournaments in $\mathcal{CC}_n, \mathcal{FC}_n$, it is obvious. For tournaments $T$ in $3\mathcal{-NCC}_n$, there are three MW pairs in $T$. By Lemma \ref{lemma:groupintersection}, they are not the same, so there are at least three teams in the union of these three pairs. By Lemma \ref{lemma:sizeofpart}, we know that there are no more than three teams in the union of these three pairs. Hence, exactly three get probability $\frac{1}{3}$, and the rest get $0$'s. For tournaments in remaining cases, the lemma follows similarly from Lemma \ref{lemma:sizeofpart} and Lemma \ref{lemma:groupintersection}.
	\qed \end{proof}
	
	We now show that for near-Condorcet tournaments, significant teams have a good chance of winning.

	\begin{lemma}\label{lemma:NWgroup}
	    Let $T$ be a near-Condorcet tournament and $r$ be the \SigOnly \newline rule. Let $S$ be an MW group in $T$. Then $r_S(T) \ge \frac{1}{2}$. Additionally, if $|S| = 2$, then $r_S(T) \ge \frac{2}{3}$.  
	\end{lemma}

	\begin{proof}
	    If $MW(i,T)$ is an MW pair, then both of its teams have assigned probability $\frac{1}{3}$, so the sum of their probabilities is $\frac{2}{3}$.
	    
	    Now assume that $MW(i,T)$ is an MW triple. Suppose that $T \in 3\mathcal{-NCP}_n$. Then there are three MW pairs in $T$. By Lemma \ref{lemma:groupintersection}, two MW pairs have to intersect at exactly one point, so even two MW pairs contain exactly three teams in their union. By Lemma \ref{lemma:sizeofpart}, these three teams are the only teams contained in the union of MW groups in $T$. Since $MW(i,T)$ is an MW triple, it contains all of them, but that is not possible since there would be four leaders in between these three teams, a contradiction with Lemma \ref{lemma:uniqueleader}.
	    
	    Suppose that $T \in 2\mathcal{-NCP}_n$. Similarly, as before, we see that there are at least three teams contained in the union of MW pairs and at most four teams in the union of all MW groups. So any $MW(i,T)$ must contain at least two teams from the union of the MW pairs. But these two teams have assigned probability $\frac{1}{3}$, so $r_{MW(i,T)}(T) \ge \frac{2}{3}$.
	    
	    Suppose that $T \in 1\mathcal{-NCP}_n$. By Lemma \ref{lemma:groupintersection}, all MW groups intersect, so $T$ contains at least one team from some MW pair. Thus, this one team has assigned probability $\frac{1}{3}$. The remaining teams are in $MW(i,T)$ (an MW triple), they have assigned probability at least $\frac{1}{9}$. Thus, $r_{MW(i,T)}(T) \ge \frac{5}{9}$.
	    
	    Suppose that $T \in 0\mathcal{-NCP}_n$. Then all teams in $MW(i,T)$ have assigned probability $\frac{1}{6}$, so $r_{MW(i,T)}(T) = \frac{1}{2}$.
	\qed \end{proof}

	\begin{corollary}\label{corollary:winninggroup}
	     Let $r$ be the \SigOnly rule, and $S$ be a group of at most three teams containing a team beating all the teams outside $S$ in $T$. Then $r_S(T) \ge \frac{1}{2}$. Additionally, if $|S| = 2$, then $r_S(T) \ge \frac{2}{3}$.
	\end{corollary}
	
	\begin{proof}
	   Follows from Lemma \ref{lemma:NWgroup} by observing that any group of three teams contains an MW pair or triple.
	\qed \end{proof}
	
Hence, every coalition capable of making a Condorcet-winner has a substantial chance of winning. We are now ready to present the more difficult parts of Theorem \ref{thm:main3snm1/2}. We need to show that no two or three teams can increase their joint probability of winning by more than $\frac{1}{3}$ or $\frac{1}{2}$, respectively.
	
	\begin{lemma}
	\label{lem:sig2snm1/3}
	The \SigOnly~rule is \kSNMx{2}{1/3} and this is tight.
	\end{lemma}

	\begin{proof}
	    Let $T$ and $T'$ be $S$-adjacent tournaments for some $S$, where $|S| = 2$. We need to show that 
	    \begin{equation} \label{eq:proveonethird}
	      r_S(T') = \sum_{i\in S} r_i(T') \le \frac{1}{3}+\sum_{i\in S}r_i(T) = \frac{1}{3}+r_S(T).  
	    \end{equation}
	    
	We will split the analysis according to the group of tournaments $T'$ belongs to. 
	By Lemma \ref{lem:properlydefined}, we know that the probabilities of all teams sum up to 1. So Condition \eqref{eq:proveonethird} is satisfied trivially when $r_S(T') \le \frac{1}{3}$ or $r_S(T) \ge \frac{2}{3}$.	
	
	Suppose that $T' \in \mathcal{CC}_n$. If $S$ does not contain the Condorcet-winner $w$ of $T'$, then $r_S(T') = 0$, and we are done. Otherwise, $w$ is in $S$ and beats all teams outside of $S$ in $T'$ as well as in $T$. Therefore, by Lemma \ref{corollary:winninggroup}, $r_S(T) \ge \frac{2}{3}$, and we are done. Now we deal with cases when $T'$ is not a Condorcet tournament. If $T$ is a Condorcet tournament, it means that its Condorcet-winner must be in $S$ (because $T'$ does not have a Condorcet-winner), and so $r_S(T) = 1$, and we are done. Thus, in the rest of the proof, we can assume that $T$ is not a Condorcet tournament. 
	
	Suppose that $T' \in \mathcal{FC}_n \cup 0\mathcal{-NCP}_n$. In this case $r$ assigns probability at most $\frac{1}{6}$ to every team in $T'$. Thus, $r_S(T')\le 2 \cdot \frac{1}{6}=\frac{2}{6}$ and we are done.
	
	Suppose that $T' \in \cup_{\ell=1}^{3} \ell\mathcal{-NCP}_n$. Let $S = \{a,b\}$. Tournament $T'$ is \newline near-Condorcet, so every team has a probability at most $\frac{1}{3}$ to win, and if it does not have probability $\frac{1}{3}$, then it has probability at most $\frac{1}{6}$. If $r_S(T')\le \frac{1}{3}$, then we are done. The only cases when it does not happen is when either both $a$ and $b$ have probability $\frac{1}{3}$ to win, or exactly one of them has. If both of them have assigned $\frac{1}{3}$, then they are from some MW pairs in $T'$. If they are from the same MW pair $MW(i,T')$, then $S$ contains this pair. Thus, by Lemma \ref{corollary:winninggroup}, $r_S(T) \ge \frac{2}{3}$ and we are done. The other option is that $a$ and $b$ belong to different MW pairs $\alpha$ and $\beta$, respectively. By Lemma \ref{lemma:groupintersection}, these pairs intersect in some team $c$ and $c$ is the leader of exactly one pair from $\alpha$ and $\beta$. Without a loss of generality, $c$ is the leader of $\alpha$. Thus, $a$ is the only team that beats $c$ in $T'$. Since $c$ is not in $S$, it is beaten by only $a$ in $T$ as well. It means that $T$ is a~near-Condorcet tournament with MW pair $\alpha$. Hence, $a$ gets probability $\frac{1}{3}$ in $T$. Thus,
	\[r_S(T')= \frac{1}{3}+\frac{1}{3} \le \frac{1}{3} + r_S(T).\]
	
	It remains to solve the case when exactly one of $a,b$ has a probability $\frac{1}{3}$ to win in $T'$. Without loss of generality, let it be $a$. Since $n\ge 6$, then $b$ get probability at most $\frac{1}{9}$. Thus, it is sufficient to show that $r_S(T) \ge \frac{1}{9}$. Let $\alpha = \{a,c\}$ be an MW pair in $T'$ team $a$ is part of. We know that $b$ is not in $\alpha$. If $a$ is not the leader of $\alpha$, by similar argument as in the previous paragraph, $r_a(T) = \frac{1}{3}$ and we are done. If $a$ is the leader of $\alpha$ then $a$ is beaten only by $c$ and $b$ in $T$. Hence, $\{a,b,c\}$ is an MW triple in $T$ with leader $a$. Either $T$ is from $2\mathcal{-NCP}_n \cup 3\mathcal{-NCP}_n$ or $1\mathcal{-NCP}_n \cup 0\mathcal{-NCP}_n$. In the latter case, all teams in MW triples have assigned probability at least $\frac{1}{9}$, so $r_S(T) \ge \frac{2}{9}$. On the other hand, if $T \in 2\mathcal{-NCP}_n \cup 3\mathcal{-NCP}_n$, then there are at least two MW pairs in $T$. By Lemma \ref{lemma:groupintersection}, these two MW pairs have to intersect at exactly one point, so there are exactly three teams in their union. By Lemma \ref{lemma:sizeofpart}, there are at most 4 teams in the union of all MW groups. Hence, $a$ or $b$ is in some MW pair and $r_S(T) \ge \frac{1}{3}$.

    The fact that this is tight follows from the lower bounds in~\cite{RSEB}. It suffices to consider a tournament with $n = 3$ teams where each team wins exactly one game. In this case, there are exactly three MW pairs. Each team will win the tournament with probability $1/3$. If any two of them collude to make one of them a Condorcet-winner, their changes of winning will increase by $1-2/3=1/3$.  
	\qed \end{proof}

	\begin{lemma}
	\label{lem:sig3snm1/2}
	    The \SigOnly~rule is \kSNMx{3}{1/2} and this is tight.
	\end{lemma}
	
	\begin{proof}
	
	Let $T$ and $T'$ be $S$-adjacent tournaments, with $|S| \leq 3$. We need to show that 
	    \begin{equation} \label{eq:proveonehalf}
	      r_S(T') = \sum_{x\in S} r_x(T') \le \frac{1}{2}+\sum_{x\in S}r_x(T) = \frac{1}{2}+r_S(T).  
	    \end{equation}
	    
	In the previous Lemma, we showed an even stronger inequality for cases when $|S| = 2$. Thus, we can assume that $|S| = 3$. We split the analysis according to the class of tournaments $T'$ belongs to. By Lemma \ref{lem:properlydefined}, we know that the probabilities of all teams sum up to 1. So Equation \eqref{eq:proveonehalf} is satisfied trivially when $r_S(T') \le \frac{1}{2}$ or $r_S(T) \ge \frac{1}{2}$.	
	
    \textbf{Case 1: $T' \in \mathcal{CC}_n$}. If $S$ does not contain the Condorcet-winner $w$ of $T'$ then $r_S(T') = 0$. Otherwise, $w$ is in $S$ and beats all teams outside of $S$, even in $T$. Therefore, by Corollary \ref{corollary:winninggroup}, $r_S(T) \ge \frac{1}{2}$. 
    
    Now we deal with cases when $T'$ is not a Condorcet tournament. If $T$ is a Condorcet tournament, it means that its Condorcet-winner must be in $S$, so $r_S(T) = 1$. Thus, in the rest of the proof, we can assume that $T$ is not a Condorcet tournament. 
	
    \textbf{Case 2: $T' \in \mathcal{FC}_n \cup 0\mathcal{-NCP}_n$}. In this case $r$ assigns probability at most $\frac{1}{6}$ to every team in $T'$. Thus, $r_S(T')\le 3 \cdot \frac{1}{6}=\frac{1}{2}$, and we are done.
	
    \textbf{Case 3: $T' \in 3\mathcal{-NCP}_n \cup 2\mathcal{-NCP}_n$}. In this case, teams in MW pairs get $\frac{1}{3}$, and the rest get zero. By Lemma \ref{lemma:sizeofpart}, there are exactly three teams with $\frac{1}{3}$ in $T'$ (there cannot be two since there are at least two different MW pairs). If $S$ contains at most one such team, then $r_S(T') \le \frac{1}{3}$ and we are done. If $S$ contains both teams from some MW pair in $T'$, then the leader of $S$ beats all the teams outside of $S$ in $T'$ as well as in $T$, so by Lemma \ref{corollary:winninggroup}, $r_S(T) \ge \frac{1}{2}$, and we are done. The remaining case is if $S$ contains exactly two teams with assigned probabilities $\frac{1}{3}$, say $a$ and $b$, and these two teams are in different MW pairs $\alpha$ and $\beta$, respectively. By Lemma \ref{lemma:groupintersection}, these pairs intersect in some team $d$, and $d$ is the leader of either $\alpha$ or $\beta$. Without loss of generality, $d$ is the leader of $\alpha$. Thus, $a$ is the only team that beats $d$ in $T'$. Since $d$ cannot be in $S$ (it would be the third team with $\frac{1}{3}$), it is beaten by $a$ in $T$ as well. It means that $T$ is a~near-Condorcet tournament with MW pair $\alpha$. Hence $a$ gets probability $\frac{1}{3}$. Thus,
		\[r_S(T')= \frac{1}{3}+\frac{1}{3} + 0 \le \frac{1}{2}+\frac{1}{3} \le \frac{1}{2} + r_S(T).\]
	
    \textbf{Case 4: $T' \in  1\mathcal{-NCP}_n$}. Let $S = \{a,b,c\}$. In this case, only  two teams in $T'$ have assigned probability $\frac{1}{3}$; the teams of the only MW pair $X$. If there are teams in MW triples, they have assigned $\frac{1}{9}$, and the rest of the probability mass is assigned uniformly to the remaining teams. It means that if there is no MW triple, the remaining teams get $\frac{1}{3(n-2)} \le \frac{1}{12}$. And if there is some MW triple, it has to contain two distinct teams from teams in $X$, so the remaining probability mass is either $\frac{1}{9}$, or $0$. Thus, in this case, every teams outside MW groups get at most $\frac{1}{9(n-4)} \le \frac{1}{18}$. We now split the analysis according to the probabilities assigned to the teams in $S$ in $T'$. If $r_S(T') \le \frac{1}{2}$, then the Equation \eqref{eq:proveonehalf} holds trivially. The only cases when this does not happen arise whenever $S$ contains two teams with $\frac{1}{3}$ or if it contains one team with $\frac{1}{3}$ and two teams with $\frac{1}{9}$.
	
	In the first case $S$ contains both teams of $X$, so once again $r_S(T) \ge \frac{1}{2}$ by Corollary \ref{corollary:winninggroup}, and Equation \eqref{eq:proveonehalf} holds. In the second case, $S$ contains one team of $X$, let say it is $a$, and two teams from some MW triples (so $r_S(T') = \frac{5}{9}$). Let us first assume that $a$ is not the leader of $X$. It means that the leader of $X$ is not in $S$ and loses only to $a$ in $T'$. Thus, he must lose only to $a$ in $T$ as well. Hence, $T$ is a~near-Condorcet tournament with MW pair $X$ and $r_S(T) \ge r_a(T) \ge \frac{1}{3}$. Therefore, Equation \eqref{eq:proveonehalf} holds. The remaining option is that $a$ is the leader of $X$. If there would be some other MW pair in $T$ other than $X$, its leader would have to be in $S$ because teams in $S$ are the only teams that can beat more teams in $T$ than in $T'$. But it would imply $r_S(T) \ge  \frac{1}{3}$, and we would be done. Hence, we can assume that there is no new MW pair in $T$. It means that all teams in MW groups get at least $\frac{1}{9}$. Thus, it is sufficient to show that $a$, $b$ or $c$ is in some MW group in $T$.

	
	Team $a$ loses only one match in $T'$. If it lost only one match in $T$, it would be the leader of an MW pair in $T$. If it lost only two matches in $T$, say to $f,g$, then it would be the leader of MW triple $\{a,f,g\}$. Hence, the only remaining case is when $a$ loses to three teams in $T$. Thus, $a$ must lose to $b$ and $c$. Team $c$ is in some MW triple $Y$ in $T'$. Let $y$ be the leader of $Y$. If $y$ is outside $S$, then $y$ still loses only two matches in $T$. Hence, $Y$ is still an MW triple in $T$ and $c$ is part of it, so we are done. The other option is that $y$ is inside $S$. It cannot be $a$, because $a$ is already a leader of $X$ in $T'$. Therefore, $y$ is either $b$, or $c$. But both $b$ and $c$ win their matches in $T$ against $a$, the matches they lost in $T'$. Hence, they both have to win at least as many matches in $T$ as in $T'$.  Thus, $y$ still loses at most two matches in $T$. Therefore, $y$ (hence $b$ or $c$) is still a leader of some MW group in $T$.

    The fact that this rule is tight follows from a slight generalization of the lower bounds of \cite{RSEB}. Suppose there are $n=6$ teams: $\{A, B, C, D, E, F\}$. Teams $A, B, C$ form a cycle. Teams $D, E, F$ also form a cycle. Then $A$ beats $E, F$ but loses to $D$, $B$ beats $D, F$ but loses to $E$ and $C$ beats $D, E$ but loses to $F$. This tournament no MW pairs, only MW triples. Moreover, each team is part of a MW triple. Therefore, the rule will assign probability $1/6$ to all teams. However, any MW triple can collude to make their leader a Condorcet-winner, increasing their chances of winning from $1/6+1/6+1/6=1/2$ to $1+0+0=1$. Therefore, our analysis of this rule is tight. 
 
	\qed \end{proof}
	Finally, we prove that the rule is monotone. 
	
	 \begin{lemma}
	 \label{lem:sigmonotone}
        The \SigOnly~rule is monotone.
    \end{lemma}

    \begin{proof}
        Let $a$ be a team and $T, T'$ be tournaments where all matches not involving team $a$ are identical and $\delta^{+}(a, T) \supseteq \delta^{+}(a, T')$. We need to show that $r_a(T) \geq r_a(T')$. By induction it is sufficient to consider only the case when $\delta^{+}(a, T) \setminus \delta^{+}(a, T')$ contains only one element, say $b$. Then $T'$ is $\{a,b\}$-adjacent tournament to $T$, and $a$ loses to $b$ in $T'$. We need to show that $r_a(T') \le r_a(T)$. Note again that only  teams $a$ and $b$ have the outcome of some of their matches different in $T'$ and $T$.
	    
	    According to the assignment of probabilities, each teams get probability $1$, $\frac{1}{3}$, $\frac{1}{6}$, $\frac{1}{9}$, 0, or uniformly distributed remaining value (that is always at most $\frac{1}{6}$). We split the analysis according to the properties of $a$ in $T$.
	    
	    \begin{itemize}
	        \item $r_a(T) = 1$; $a$ is a Condorcet-winner:\\
	        In this case, $r_a(T')$ cannot be bigger than 1. 
	        
	        \item $r_a(T) = \frac{1}{3}$; $a$ is part of some MW pair in near-Condorcet tournament $T$:\\
	        Team $a$ loses one more match in $T'$ than in $T$. Hence it cannot be the Condorcet-winner in $T'$. Thus, its probability could not increase.
	        
	        \item $r_a(T) = \frac{1}{6}$; $T$ is a near-Condorcet tournament without MW pairs and $a$ is part of some MW triple:\\
	    Similarly as before, $a$ cannot become the Condorcet-winner.
     
         Only $b$ wins more matches in $T'$. Thus, only $b$ can possibly be a leader of an MW pair in $T'$. But $a$ loses to $b$ in $T'$. Hence, $a$ cannot be part of this possible new pair. Hence $r_a(T')$ cannot be $1$ or $\frac{1}{3}$ and thus it is at most $\frac{1}{6}$.
	        
	        \item $r_a(T) = \frac{1}{9}$; $T$ is a near-Condorcet tournament with exactly one MW pair \newline $MW(i,T)$, and $a$ is a part of some MW triple $MW(j, T)$:\\
	        By the same argument as before, $a$ cannot become Condorcet-winner nor a part of some MW pair. Additionally, $MW(i,T)$ still remains an MW pair in $T'$, so all teams in MW triples in $T'$ get at most $\frac{1}{9}$, and not significant teams get at most $\frac{1-\frac{2}{3}}{n-2} \le \frac{1}{12}$ (there are at least 6 teams, and the two in $MW(i,T)$ get $\frac{1}{3}$). Thus, the value of $r_a(T')$ is at most $\frac{1}{9}$.

	        \item $r_a(T) = 0$; $T$ is a near-Condorcet tournament with exactly two MW pairs, and $a$ is part of an MW triple:\\
	       There must be at least three teams in the union of the two MW pairs. Since $a$ is not a part of these pairs, these teams also form the same MW pairs in $T'$. Since all teams in MW pairs have assigned probability $\frac{1}{3}$, their joint probability to win is 1. Thus, team $a$ must have assigned $0$ in both $T$ and $T'$.

               \item $a$ is not significant:\\
	       No team can become a leader of a new MW pair because only team $b$ beats more teams in $T'$ than in $T$, and since $a$ is not in any MW group in $T$, $b$ could not lose exactly two matches in $T$.
	        It means that all teams in MW pairs in $T$ get the same probability as teams in MW pairs in $T'$, and the same holds for MW triples. Furthermore, all MW groups in $T$ are still MW groups in $T'$ because $a$ is not part of them. It means that in $T$ we have a group $S$ of all significant teams and a group of remaining teams that have the remaining probability assigned uniformly. Denote this uniform value by $x$. Note that this $x$ is smaller or equal to the significant teams' values. Additionally, note that $a$ gets $x$ since it is not significant. In tournament $T'$, teams from $S$ get the same probabilities as in $T$. Furthermore, a new MW triple can be formed there. Teams inside this new triple get probabilities at least $x$; more than in $T$. And the remaining teams in $T'$ get the remaining probability distributed uniformly; denote it by $y$. Since some teams could become significant and gotten more, $y\le x$. Note that $a$ could not become significant, because $b$ can be the only new leader and $a$ loses to $b$ in $T'$. Therefore, $a$ got $y$ in $T'$, and its assigned probability could not increase.

	    \end{itemize}
	\qed \end{proof}
    
    We are now ready to prove the main result of Section~\ref{sec:sigOnly}.

     \begin{proof}[Proof of Theorem~\ref{thm:main3snm1/2}]
    Follows from Lemmas~\ref{lem:properlydefined},~\ref{lem:sig2snm1/3}, \ref{lem:sig3snm1/2}, ~\ref{lem:sigmonotone} and the observation that the \SigOnly rule is Condorcet-consistent by construction. 
    \qed \end{proof}

\section{Missing Proofs from Section~\ref{sec:reduction}}
\label{app:reduction}

Before we proceed, we need to introduce definitions and tournament rules only relevant to Section~\ref{sec:reduction}. 

\begin{definition}
\label{def:top-cycle}
Given a tournament graph $T$, the \emph{top cycle} $\mathcal{C}(T)$ is the minimal non-empty set of teams $S$ such that there are no edges $(i, j)$ with $i \not \in S$ and $j \in S$. A rule is \emph{top-cycle consistent} if it always selects a team in the top-cycle, i.e. $r_{\mathcal{C}(T)}(T) = 1$.
\end{definition}

\begin{lemma}\cite{LPmanipulability}
Any top-cycle consistent rule $r$ is also Condorcet-consistent. 
\end{lemma}

We are now ready to introduce our simple \textsc{TopCycle} Rule. 

\begin{definition}[\textsc{TopCycle} Rule]
\label{def:topcycle} The \emph{Uniform Top-Cycle} $\textsc{TopCycle}$ rule operates as follows. Given tournament $T$, sample the winner uniformly at random a team $i \in C(T)$. 
\end{definition}

The proof of Theorem~\ref{thm:inc_num_of_teams_arbitrarily_top-cycle} will be incremental. We first prove how to extend rules $r$ for $n$ teams to rules $r'$ for $n^d$ for any $d$. Consider the following rule $r' := \textsc{Ext}(r, n^d)$ and call it \emph{the extension of} $r$ to $n^d$ teams. Given a tournament $T$ on a~set~$N$ of $n^d$ teams, $r'$ assigns to~$T$ a~distribution over winners obtained in the following three steps:
    \begin{enumerate}[(i)]
        \item Choose a~permutation~$\pi$ on~$N$ uniformly at random. Partition~$N$ into~$n$ groups $G_1^{\pi}, G_2^{\pi}, \dots, G_n^{\pi}$ each containing~$n^{d-1}$ teams such that \[\forall t \in [n] \ : \ G_t^{\pi} = \{\pi(i) \mid (t-1) \cdot n^{d-1} + 1 \leq i \leq t \cdot n^{d-1}\} .\]
        \item For every $t \in [n]$, play the sub-tournament of ~$T$ induced by teams $G^{\pi}_t$ using \textsc{TopCycle} rule to obtain a~distribution $\Delta_t^{\pi}$ over winners in~$G^{\pi}_t$.
        \item For every $t \in [n]$, select a~finalist~$w_t$ from $G^{\pi}_t$ according to the distribution~$\Delta_t^{\pi}$. Let $W := \{w_t \mid t \in [n]\}$ be the set of finalists and play the sub-tournament~$T$ over teams in $W$ using~the rule~$r$ to obtain a~distribution $\Delta_W^\pi$ over winners in $W$.
    \end{enumerate}
    The steps above are a formalization of the rule described in Section~\ref{sec:reduction}. We first bound the manipulability of rule $r'$ as a function of the manipulability of $r$. 

    \begin{theorem} \label{thm__inc_num_of_teams_arbitrarily}
        Let $d, k \geq 2$.
        If there exists a~CC and $\kSNMx{k}{\alpha}$ $r$ for $n$ teams, then there exists a~CC and $\kSNMx{k}{\alpha'}$ $r'$ for $n^d$ teams with 
        \[\alpha' \leq \alpha \Big( 1 - \frac{(k-1)^2}{n} \Big) + \frac{(k-1)^2}{n}.\]
    \end{theorem}

	\begin{proof}[Proof of Theorem~\ref{thm__inc_num_of_teams_arbitrarily}]
	    Suppose that there exists a~CC and $\kSNMx{k}{\alpha}$~rule $r$ for $n$ teams. We will let $r'$ be the extension of $r$ to $n^d$ teams. Since $\textsc{TopCycle}$ is CC and we assumed that $r$ is CC, then it follows that $r'$ is also CC.

	    We now try to find~$\alpha'$ as small as possible so that~$r'$ is $\kSNMx{k}{\alpha'}$. We proceed directly from the definition:
	    Let~$T$ be a~tournament on a~set~$N$ of~$n^d$ teams and let $S \subseteq N$ be a~set of colluding teams. Assume wlog that $|S| = k.$ 
	    For any tournament $T'$ which is $S$-adjacent to $T$, we show that $r'_S(T') - r'_S(T) \leq \alpha'$.

	   For a~permutation $\pi \in \mathcal{S}_{n^d}$, where $\mathcal{S}_{n^d}$ is the set of all (labeled) seedings of the set of teams $N$, let~$A_\pi$ be an event that~$r'$ chooses~$\pi$ in the first step.
        Let~$B_T$ and~$B_{T'}$ be two events that the winner of~$T$, respectively~$T'$, selected by~$r'$ is in~$S$.
        Now, we can write:
	    \begin{align*}
	        r'_S(T') - r'_S(T) = \pr{B_{T'}} - \pr{B_T} \\
	            = \sum_{\pi \in \Sn{n^d}} \Big( \cpr{B_{T'}}{A_\pi} \cdot \pr{A_\pi} - \cpr{B_T}{A_\pi} \cdot \pr{A_\pi} \Big) \\
	            = \sum_{\pi \in \Sn{n^d}} \Big( \cpr{B_{T'}}{A_\pi} \cdot \frac{1}{(n^d)!} - \pr{B_T}{A_\pi} \cdot \frac{1}{(n^d)!} \Big) \\
	            = \frac{1}{(n^d)!} \cdot \sum_{\pi \in \Sn{n^d}} \Big( \cpr{B_{T'}}{A_\pi}  - \cpr{B_T}{A_\pi} \Big).
	    \end{align*}
	    Let~$\Sn{n^d}^{+}$ be the set of all permutation $\pi \in \Sn{n^d}$ such that $\forall t \in [n] \ : \ |G_t^\pi \cap S| \leq 1$ and let $\Sn{n^d}^{-} := \Sn{n^d} \setminus \Sn{n^d}^{+}$.
	    %
	    %
	    We can continue to write:
        \begin{align*}
            r'_S(T') - r'_S(T) &= \frac{1}{(n^d)!} \cdot \sum_{\pi \in \Sn{n^d}^{+}} \Big( \cpr{B_{T'}}{A_\pi} - \cpr{B_T}{A_\pi} \Big) \\
                &+ \frac{1}{(n^d)!} \cdot \sum_{\pi \in \Sn{n^d}^{-}} \Big( \cpr{B_{T'}}{A_\pi} - \cpr{B_T}{A_\pi} \Big) \\
                &\leq \frac{1}{(n^d)!} \cdot \sum_{\pi \in \Sn{n^d}^{+}} \Big( \cpr{B_{T'}}{A_\pi} - \cpr{B_T}{A_\pi} \Big) \\
                &+ \frac{1}{(n^d)!} \cdot \sum_{\pi \in \Sn{n^d}^{-}} 1.
        \end{align*}
        For every $\pi \in \Sn{n^d}^{+}$, our goal is to upper bound $\cpr{B_{T'}}{A_\pi} - \cpr{B_T}{A_\pi}$ by~$\alpha$. For that we use again conditional probability. 
        Given a~permutation $\pi \in \Sn{n^d}^{+}$ and a~tournament~$T$, let $C_{\pi}^{T}(w_1, w_2, \dots, w_n)$ be an event that~$r'$ chooses~$\pi$ in the first step (i.e., $A_\pi$ happens) and~$r'$ chooses $(w_1, w_2, ..., w_n) \in G_{1}^{\pi} \times G_{2}^{\pi} \times \cdots \times G_{n}^{\pi}$ in the third step. Then
	    %
        %
        \begin{align*}
	        \cpr{B_{T'}}{A_\pi} - \cpr{B_T}{A_\pi} \\
	        = \sum \cpr{B_{T'}}{C_{\pi}^{T'}(w_1, \dots, w_n)} \cdot \pr{C_{\pi}^{T'}(w_1, \dots, w_n)} \\
	        - \sum \cpr{B_{T}}{C_{\pi}^{T}(w_1, \dots, w_n)} \cdot \pr{C_{\pi}^{T}(w_1, \dots, w_n)} \\
	        = \sum \Big( \cpr{B_{T'}}{C_{\pi}^{T'}(w_1, \dots, w_n)} \cdot \pr{C_{\pi}^{T'}(w_1, \dots, w_n)} \\
	        - \cpr{B_{T}}{C_{\pi}^{T}(w_1, \dots, w_n)} \cdot \pr{C_{\pi}^{T}(w_1, \dots, w_n)} \Big),
	    \end{align*}
	    where the sums are over $(w_1, w_2, ..., w_n) \in G_{1}^{\pi} \times G_{2}^{\pi} \times \cdots \times G_{n}^{\pi}$.
	    Since $\pi \in \Sn{n^d}^{+}$, we have that $T[G^{\pi}_t] = T'[G^{\pi}_{t}]$. Hence $\pr{C_{\pi}^{T'}(w_1, \dots, w_n)} = \pr{C_{\pi}^{T}(w_1, \dots, w_n)}$.
	    Moreover, $$\cpr{B_{T'}}{C_{\pi}^{T'}(w_1, \dots, w_n)} - \cpr{B_{T}}{C_{\pi}^{T}(w_1, \dots, w_n)} \leq \alpha$$ because~$r$ is $\kSNMx{k}{\alpha}$. Thus, $$\cpr{B_{T'}}{A_\pi} - \cpr{B_T}{A_\pi} \leq \alpha.$$ 
	    And so, 
	    \begin{align*}
        r'_S(T') - r'_S(T) &\leq \frac{1}{(n^d)!} \cdot \sum_{\pi \in \Sn{n^d}^{+}} \alpha + \frac{1}{(n^d)!} \cdot \sum_{\pi \in \Sn{n^d}^{-}} 1 \\
            &= \alpha \cdot \frac{|\Sn{n^d}^{+}|}{(n^d)!} + \frac{|\Sn{n^d}^{-}|}{(n^d)!} \\
            &= \alpha \cdot \frac{|\Sn{n^d}^{+}|}{(n^d)!} + \frac{(n^d)! - |\Sn{n^d}^{+}|}{(n^d)!} \\
            &= \alpha \cdot \frac{|\Sn{n^d}^{+}|}{(n^d)!} + 1 - \frac{|\Sn{n^d}^{+}|}{(n^d)!} \\
            &= 1 - \frac{|\Sn{n^d}^{+}|}{(n^d)!} \cdot (1 - \alpha)
	    \end{align*}
	    Observe that to upper bound the last expression is sufficient to lower bound~$\frac{|\Sn{n^d}^{+}|}{(n^d)!}$.
	    Also observe that
	    \begin{align*}
	        |\Sn{n^d}^{+}| = n^{d} \cdot (n^{d} - n^{d-1}) \cdot \dots \cdot (n^{d} - (k-1)n^{d-1}) \cdot (n^d - k)!.
	    \end{align*}
	    Hence
	    \begin{align*}
	        \frac{|\Sn{n^d}^{+}|}{(n^d)!} &= \frac{n^{d} \cdot (n^{d} - n^{d-1}) \cdot \dots \cdot (n^{d} - (k-1)n^{d-1}) \cdot (n^d - k)!}{(n^d)!} \\
	            &= \frac{(n^{d} - n^{d-1}) \cdot \dots \cdot (n^{d} - (k-1)n^{d-1})}{(n^{d}-1) \cdot \dots \cdot (n^{d} - (k - 1))} \\
	            &= \frac{n^{k (d-1)} \cdot (n - 1) \cdot \dots \cdot (n - (k-1))}{n^{k (d-1)} \cdot (n-1/n) \cdot \dots \cdot (n - (k - 1)/n)} \\
	            &\geq \frac{(n - (k-1))^{k-1}}{n^{k-1}} = \bigg( 1 - \frac{k-1}{n} \bigg)^{k-1} \\
	            &\geq 1 - \frac{(k-1)^2}{n}.
	    \end{align*}
	    Therefore,
	    \begin{align*}
	        r'_S(T') - r'_S(T) &\leq 1 - \frac{|\Sn{n^d}^{+}|}{(n^d)!} \cdot (1 - \alpha) \\
	            &\leq 1 - \bigg( 1 - \frac{(k-1)^2}{n} \bigg)(1 - \alpha) \\
	            &= \alpha + \frac{(k-1)^2}{n} - \alpha \cdot \frac{(k-1)^2}{n} \\
	            &= \alpha \Big( 1 - \frac{(k-1)^2}{n} \Big) + \frac{(k-1)^2}{n}.
	    \end{align*}
    \qed \end{proof}
    
    Next we prove a simple lemma, which states that the extension of any top-cycle consistent rule $r$ remains top-cycle consistent. 
    
     \begin{lemma} \label{lemma__extension_preserve_top_cycle}
        If rule~$r$ for~$n$ teams is top-cycle consistent, then any extension to $n^d$ teams $r' = \textsc{Ext}(r, n^d)$ is top-cycle consistent as well.
    \end{lemma}
	    
    \begin{proof} [Proof of Lemma~\ref{lemma__extension_preserve_top_cycle}]
        Let~$T$ be a~tournament on a~set~$N$ of~$n^{d}$ teams.
        We observe that if~$T'$ is a~sub-tournament of~$T$ such that $V(T') \cap \mathcal{C}(T) \neq \emptyset$, then $\mathcal{C}(T') \subseteq \mathcal{C}(T)$.
        Indeed, since every team in $V(T') \setminus \mathcal{C}(T)$ is beaten by every team in $x \in V(T') \cap \mathcal{C}(T)$, the top-cycle~$\mathcal{C}(T')$ in~$T'$ is a~subset of $V(T') \cap \mathcal{C}(T) \subseteq \mathcal{C}(T)$.
	 
        To prove that~$r'$ is top-cycle consistent, we need to show that $r'_{x}(T) = 0$ for every $x \in N \setminus \mathcal{C}(T)$.
        We use the notation introduced in the description of the extension of rule~$r'$.
        It is enough to show that $\Delta^{\pi}_{W}(x) = 0$ for any permutation~$\pi$ on~$N$.
        Let $t \in [n]$ be such that $x \in G^{\pi}_{t}$. We consider two cases.
	    
        First, suppose that $G^{\pi}_{t} \cap \mathcal{C}(T) \neq \emptyset$. Let $T[G^{\pi}_{t}]$ be the sub-tournament of $T$ induced by $G^{\pi}_{t}$. From the observation follows that $\mathcal{C}(T[G^{\pi}_{t}]) \subseteq \mathcal{C}(T)$. 
        Since~\textsc{TopCycle} is top-cycle consistent, we have $w_t \in G^{\pi}_{t} \cap \mathcal{C}(T)$ (i.e., the winner of $T[G^{\pi}_{t}]$ is in $G^{\pi}_{t} \cap \mathcal{C}(T)$) with probability~1.
        Thus, $\Delta^{\pi}_{W}(x) = 0$.
	    
        Second, suppose that $G^{\pi}_{t} \cap \mathcal{C}(T) = \emptyset$. We claim that         $\Delta^{\pi}_{W}(w_t) = 0$ for any choice of~$w_t \in G^{\pi}_{t} \cap \mathcal{C}(T)$.
        Let $p \in[n]$ be such that $G_p^\pi \cap \mathcal{C}(T) \neq \emptyset$.
        Since~\textsc{TopCycle} is top-cycle consistent, we have that $w_p \in G_p^\pi \cap \mathcal{C}(T)$ with probability~1.
        Hence $\mathcal{C}(T) \cap W \neq \emptyset$ and $\mathcal{C}(T[W]) \subseteq \mathcal{C}(T)$ with probability~1.
        And since~$r$ is top-cycle consistent, the winner of $T[W]$ is in $\mathcal{C}(T)$ with probability~1. 
        Thus, $\Delta^{\pi}_{W}(w_t) = 0$ and so $\Delta^{\pi}_{W}(x) = 0$.
    \qed \end{proof}
	    
We are now ready to prove Theorem~\ref{thm:inc_num_of_teams_arbitrarily_top-cycle}.

\begin{proof}[Proof of Theorem~\ref{thm:inc_num_of_teams_arbitrarily_top-cycle}]
Suppose that there exists a~top-cycle consistent and \newline $\kSNMx{k}{\alpha}$ rule $r$ for $n$ teams. Let~$d$ be the smallest natural number $n' \leq n^d$. Given a tournament $T$ on a set $N'$ of $n'$ teams, extend the set of teams by adding a~set~$D$ of $n^d-n'$ many dummy teams (who lose to every real team, the outcomes inside of~$D$ are irrelevant) to obtain a tournament~$T'$ on $N' \cup D$ of~$n^d$ teams. Observe that this operation preserves the top-cycle, i.e., $\mathcal{C}(T) = \mathcal{C}(T')$. By Lemma~\ref{lemma__extension_preserve_top_cycle} and Theorem~\ref{thm__inc_num_of_teams_arbitrarily}, we know that~$\textsc{Ext}(r, n^d)$ is top-cycle consistent and $\kSNMx{k}{\alpha'}$ for some \[\alpha' \leq \alpha \Big( 1 - \frac{(k-1)^2}{n} \Big) + \frac{(k-1)^2}{n}.\]
Finally, the rule~$r'$ for $T$ is defined as follows: $r'_x(T) := \textsc{Ext}(r,n^d)_x(T')$ for every $x \in N'$. The rule~$r'$ is well defined because $\mathcal{C}(T) = \mathcal{C}(T')$ and~$\textsc{Ext}(r,n^d)$ is top-cycle consistent. Moreover, observe that~$r'$ is also top-cycle consistent and $\kSNMx{k}{\alpha'}$.

    \qed \end{proof}

\end{document}